\documentclass{eptcs}

\usepackage{textcomp}
\usepackage{multirow}
\usepackage{caption}
\usepackage{xspace}
\usepackage{amssymb}
\usepackage{rotating}
\usepackage{xcolor}
\usepackage{subcaption}

\makeatletter
\def\moverlay{\mathpalette\mov@rlay}
\def\mov@rlay#1#2{\leavevmode\vtop{%
   \baselineskip\z@skip \lineskiplimit-\maxdimen
   \ialign{\hfil$\m@th#1##$\hfil\cr#2\crcr}}}
\newcommand{\charfusion}[3][\mathord]{
    #1{\ifx#1\mathop\vphantom{#2}\fi
        \mathpalette\mov@rlay{#2\cr#3}
      }
    \ifx#1\mathop\expandafter\displaylimits\fi}
\makeatother

\usepackage{graphicx}
\usepackage{algpseudocode}
\usepackage{algorithm}
\usepackage{booktabs}
\usepackage{hyperref}
\usepackage{amsmath}
\usepackage{amsthm}
\usepackage[mode=buildnew]{standalone}
\usepackage{tikz}
\definecolor{darkgreen}{RGB}{68,180,46}
\definecolor{darkgray}{RGB}{80,80,80}
\definecolor{aauBlue}{RGB}{85,158,195}
\definecolor{aauBlue2}{RGB}{67,111,129}
\usepackage{todonotes}
\newcommand{\Oprime}{\ensuremath{\mathcal{O}^{'}}\xspace}

\newcommand{\Opprime}{\ensuremath{\mathcal{O}^{''}}\xspace}
\newcommand{\Mo}{\ensuremath{\mathcal{O}}\xspace}

\newcommand{\AT}{\ensuremath{A\!-\!T}}
\newcommand{\NATCAT}{\ensuremath{NAT\!-\!CAT}}
\newcommand{\NATCACT}{\ensuremath{NAT\!-\!CACT}}
\newcommand{\EAT}{\ensuremath{E\!-\!AT}}

\newcommand{\Rql}{\ensuremath{\mathcal{R}^{ql}}}
\newcommand{\IR}[1]{\ensuremath{\pi(#1)}}
\newcommand{\ontotrans}[1]{\ensuremath{\tau(#1)}}
\newcommand{\querytrans}[1]{\ensuremath{\tau^q(#1)}}

\newcommand{\Qcal}{\ensuremath{\mathcal{Q}}}

\newcommand{\Tcal}{\ensuremath{\mathcal{T}}}
\newcommand{\Acal}{\ensuremath{\mathcal{A}}}
\newcommand{\Pcal}{\ensuremath{\mathcal{P}}}
\newcommand{\Kcal}{\ensuremath{\mathcal{K}}}

\newcommand{\AN}{\ensuremath{\Acal^N}}
\newcommand{\AC}{\ensuremath{\Acal^C}}

\newcommand{\TC}{\ensuremath{\Tcal^C}}

\newcommand{\pair}[2]{\ensuremath{\langle #1, #2 \rangle}}

\newcommand{\Sig}[1]{\ensuremath{\mathbf{S}(#1)}}
\newcommand{\LME}{\ensuremath{LM\!E}}

\newcommand{\K}[2]{\ensuremath{\Kcal^{#1}_{#2}}}
\newcommand{\KAT}[1]{\ensuremath{\K{\AT}{}(#1)}}
\newcommand{\KNATCAT}[1]{\ensuremath{\K{\NATCAT}{}(#1)}}
\newcommand{\KNATCACT}[1]{\ensuremath{\K{\NATCACT}{}(#1)}}
\newcommand{\KEAT}[1]{\ensuremath{\K{\EAT}{}(#1)}}
\newcommand{\KAll}[1]{\ensuremath{\K{}{\All}(#1)}}
\newcommand{\KMod}[1]{\ensuremath{\K{}{\Mod}(#1)}}

\newcommand{\All}{\ensuremath{All}}
\newcommand{\Mod}{\ensuremath{Mod}}

\newcommand{\nop}[1]{}

\newtheorem{definition}{Definition}
\newtheorem{theorem}{Theorem}
\newtheorem{proposition}{Proposition}

\title{Efficient OWL2QL Meta-reasoning Using ASP-based Hybrid Knowledge Bases}

 \author{Haya Majid Qureshi \qquad\qquad Wolfgang Faber
   \institute{University of Klagenfurt\\ Austria}
   \email{\{haya.qureshi,wolfgang.faber\}@aau.at}
}

\def\titlerunning{Efficient OWL2QL Meta-reasoning Using ASP-based Hybrid Knowledge Bases}
\def\authorrunning{H.M. Qureshi, W. Faber} 
\begin{document}
\maketitle

\begin{abstract}
  Metamodeling refers to scenarios in ontologies in which classes and roles can be members of classes or occur in roles. This is a desirable modelling feature in several applications, but allowing it without restrictions is problematic for several reasons, mainly because it causes undecidability. Therefore, practical languages either forbid metamodeling explicitly or treat occurrences of classes as instances to be semantically different from other occurrences, thereby not allowing metamodeling semantically. Several extensions have been proposed to provide metamodeling to some extent. Building on earlier work that reduces metamodeling query answering to Datalog query answering, recently reductions to query answering over hybrid knowledge bases were proposed with the aim of using the Datalog transformation only where necessary. Preliminary work showed that the approach works, but the hoped-for performance improvements were not observed yet. In this work we expand on this body of work by improving the theoretical basis of the reductions and by using alternative tools that show competitive performance.
\end{abstract}


\section{Introduction}
Metamodeling helps in specifying conceptual modelling requirements with the notion of meta-classes (for instance, classes that are instances of other classes) and meta-properties (relations between meta-concepts). These notions can be expressed in OWL Full. However, OWL Full is so expressive for metamodeling that it leads to undecidability \cite{motik2005properties}. 
OWL 2 DL and its sub-profiles guarantee decidability, but they provide a very restricted form of metamodeling \cite{hitzler2009foundations} and give no semantic support due to the
prevalent Direct Semantics (DS).

Consider an example adapted from \cite{guizzardi2015towards}, concerning the modeling of biological species, stating that all golden eagles are eagles, all eagles are birds, and Harry is an instance of GoldenEagle, which further can be inferred as an instance of Eagle and Bird. However, in the species domain one can not just express properties of and relationships among species, but also express properties of the species themselves. For example ``GoldenEagle is listed in the IUCN Red List of endangered species'' states that GoldenEagle as a whole class is an endangered species. Note that this is also not a subclass relation, as Harry is not an endangered species. To formally model this expression, we can declare GoldenEagle to be an instance of new class EndangeredSpecies.
\begin{quote} \itshape
Eagle $\sqsubseteq$ Bird, $\>\>\>$ GoldenEagle $\sqsubseteq$ Eagle, $\>\>\>$ GoldenEagle(Harry)\\
EndangeredSpecies $\sqsubseteq$ Species, EndangeredSpecies(GoldenEagle)
\end{quote}

Note that the two occurrences of the IRI for GoldenEagle (in a class position and in an individual position) are treated as different objects in the standard direct semantics \emph{DS}\footnote{http://www.w3.org/TR/2004/REC-owl-semantics-20040210/}, therefore not giving semantic support to punned\footnote{http://www.w3.org/2007/OWL/wiki/Punning} entities and treating them as independent of each other by reasoners. These restrictions significantly limit meta-querying as well, since the underlying semantics for SPARQL queries over OWL 2 QL is defined by the \emph{Direct Semantic Entailment Regime} \cite{glimm2011using}, which uses \emph{DS}.

To remedy the limitation of metamodeling, Higher-Order Semantics (HOS) was introduced in \cite{lenzerini2015higher} for OWL 2 QL ontologies and later referred to as Meta-modeling Semantics (MS) in \cite{lenzerini2020metaquerying}, which is the terminology that we will adopt in this paper. The interpretation structure of HOS follows the Hilog-style semantics of \cite{chen1993hilog}, which allows the elements in the domain to have polymorphic characteristics. Furthermore, to remedy the limitation of metaquerying, the Meta-modeling Semantics Entailment Regime (MSER) was proposed in \cite{cima2017sparql}, which does allow meta-modeling and meta-querying using SPARQL by reduction from query-answering over OWL 2 QL to Datalog queries.


In \cite{QureshiF23} several methods were proposed that reduce query-answering over OWL 2 QL to queries over hybrid knowledge bases instead. The idea there was to split the input ontology into two parts, one involving metamodeling and one that does not. The former is transformed to Datalog using the method of \cite{cima2017sparql}, while the latter is kept as an ontology and linked to the Datalog program. The precise bridge rules to be created were either all possible or just those relevant to the query (using an established module notion). Experiments using \emph{HEXLite-owl-api-plugin} as a hybrid reasoner showed this to be a viable approach, even if the observed performance was not as quick as hoped for. This appeared to be due to internals of the hybrid reasoner and the lack of any query-oriented optimisations such as the magic set technique. Indeed, results in \cite{qureshi2021evaluation} indicate that absence of a query-oriented method is detrimental for performance.

In this work, we first recall the methods introduced in \cite{QureshiF23}, then provide a detailed proof of correctness, and, most importantly,
we use an extension of \textit{DLV2} with \textit{Python external atoms} as a hybrid reasoner. The system does support the magic set technique and our experiments show much better performance using this system.

\section{Preliminaries}
This section gives a brief overview of the language and the formalism used in this work.

\subsection{OWL 2 QL}
This section recalls the syntax of the ontology language OWL 2 QL and the \textit{Metamodeling Semantics} (MS) for OWL 2 QL, as given in \cite{lenzerini2021metamodeling}.

\subsubsection{\textbf{Syntax}}
We start by recalling some basic elements used for representing knowledge in ontologies: \textit{Concepts}, a set of individuals with common properties, \textit{Individuals}, objects of a domain of discourse, and \textit{Roles}, a set of relations that link individuals. An OWL 2 ontology is a set of axioms that describes the domain of interest. The elements are classified into \textit{literals} and \textit{entities}, where \textit{literals} are values belonging to datatypes and \textit{entities} are the basic ontology elements denoted by \textit{Internationalized Resource Identifiers} (IRI). The notion of the vocabulary $V$ of an OWL 2 QL, constituted by the tuple $V= (V_{e}, V_{c}, V_{p}, V_{d}, D, V_{i}, L_{QL})$. In $V$, $V_{e}$ is the union of $ V_{c}, V_{p}, V_{d}, V_{i}$ and its elements are called atomic expressions; $V_{c}, V_{p}, V_{d},$ and $V_{i}$ are sets of IRIs, denoting, respectively, classes, object properties, data properties, and individuals, $L_{QL}$ denotes the set of literals - characterized as OWL 2 QL datatype maps denoted as $DM_{QL}$ and $D$ is the set of datatypes in OWL 2 QL (including rdfs:Literal). Given a vocabulary $V$ of an ontology $\mathcal{O}$, we denote by $Exp$ the set of well formed expressions over $V$. For the sake of simplicity we use Description Logic (DL) syntax for denoting expressions in OWL 2 QL. Complex expressions are built over $V$, for instance, if $e_{1},e_{2} \in V$ then $\exists e_{1}.e_{2}$ is a complex expression. 
An OWL 2 QL Knowledge Base $\mathcal{O}$ is a pair $\langle \mathcal{T},\mathcal{A} \rangle$, where $\mathcal{T}$ is the TBox (inclusion axioms)  and $\mathcal{A}$ is the ABox (assertional axioms). Sometimes we also let $\mathcal{O}$ denote $\mathcal{T}\cup \mathcal{A}$ for simplicity. OWL 2 QL is a finite set of logical axioms. The axioms allowed in an OWL 2 QL ontology have one of the forms: inclusion axioms $e_{1} \sqsubseteq e_{2}$, disjointness axioms $e_{1} \sqsubseteq \neg\> e_{2}$, axioms asserting property i.e., reflexive property $ref(e)$ and irreflexive property $irref(e)$ and assertional axioms i.e., $c(a)$ class assertion, $,p(a,b)$ object property assertion, and $d(a,b)$ data property assertion. We employ the following naming schemes (possibly adding subscripts if necessary): c,p,d,t denote a class, object property, data property and datatype. The above axiom list is divided into TBox axioms (further divided into positive TBox axioms and negative TBox axioms) and ABox axioms. The positive TBox axioms consist of all the inclusion and reflexivity axioms, the negative TBox axioms consist of all the disjointness and irreflexivity axioms and ABox consist of all the assertional axioms. For simplicity, we omit OWL 2 QL axioms that can be expressed by appropriate combinations of the axioms specified in the above axiom list. Also, for simplicity we assume to deal with ontologies containing no data properties. 

\subsubsection{\textbf{Meta-modeling Semantics}}
The Meta-modeling Semantics (MS) is based on the idea that every entity in $V$ may simultaneously have more than one type, so it can be a class, or an individual, or data property, or an object property or a data type. To formalise this idea, the Meta-modeling Semantics has been defined for OWL 2 QL. In what follows, $\mathbf{P}(S)$ denotes the power set of $S$. The meta-modeling semantics for $\mathcal{O}$ over $V$ is based on the notion of interpretation, constituted by a tuple $\mathcal{I} = \langle \Delta, \cdot^I, \cdot^C, \cdot^P, \cdot^D, \cdot^T,\cdot^\mathcal{I}\rangle$, where
\begin{itemize}
    \item $\Delta$ is the union of the two non-empty disjoint sets: $\Delta = \Delta^{o} \cup \Delta^{v}$, where $\Delta^o$ is the object domain, and $\Delta^{v}$ is the value domain defined by $DM_{QL}$;
    \item $\cdot^I : \Delta^o \to \{True, False\}$ is a total function for each object $o \in \Delta^o$, which indicates whether $o$ is an individual; if $\cdot^C, \cdot^P, \cdot^D, \cdot^T$ are undefined for an $o$, then we require $o^{I} = True$, also in other cases, e.g.,\ if $o$ is in the range of $\cdot^C$;
    \item $\cdot^C : \Delta^o \to \mathbf{P}(\Delta^o)$ is partial and can assign the extension of a class;
    \item $\cdot^P : \Delta^o \to \mathbf{P}(\Delta^o \times \Delta^o)$ is partial and can assign the extension of an object property;
    \item $\cdot^D : \Delta^o \to \mathbf{P}(\Delta^o \times \Delta^{v})$ is partial and can assign the extension of a data property;
    \item $\cdot^T : \Delta^o \to \mathbf{P}(\Delta^{v})$ is partial and can assign the extension of a datatype;
    \item $^{.\mathcal{I}}$ is a function that maps every expression in $Exp$ to $\Delta^o$ and every literal to $\Delta_{v}$.
\end{itemize}

This allows for a single object $o$ to be simultaneously interpreted as an individual via $^.{}^I$, a class via $^.{}^C$, an object property via $^.{}^P$, a data property via $^.{}^D$, and a data type via $^.{}^T$. For instance, for Example 1, $\cdot^C,\>\cdot^I$ would be defined for \textit{GoldenEagle}, while $\cdot^{P},\cdot^{D}$ and $\cdot^{T}$ would be undefined for it.

The semantics of logical axiom $\alpha$ is defined in accordance with the notion of axiom satisfaction for an MS interpretation $\mathcal{I}$. The complete set of notions is specified in Table 3.B in \cite{lenzerini2021metamodeling}. Moreover, $\mathcal{I}$ is said to be a model of an ontology $\mathcal{O}$ if it satisfies all axioms of $\mathcal{O}$. Finally, an axiom $\alpha$ is said to be logically implied by $\mathcal{O}$, denoted as $\mathcal{O} \models  \alpha $, if it is satisfied by every model of $\mathcal{O}$.

\subsection{Hybrid Knowledge Bases}
Hybrid Knowledge Bases ($\mathrm{HKBs}$) have been proposed for coupling logic programming (LP) and Description Logic (DL) reasoning on a clear semantic basis. Our approach uses $\mathrm{HKBs}$ of the form $\Kcal=\pair{\Mo}{\Pcal}$, where \Mo is an OWL 2 QL knowledge base and $\mathcal{P}$ is a hex program, as defined next.

Hex programs \cite{eiter2016model} extend answer set programs with external computation sources. We use hex programs with unidirectional external atoms, which import elements from the ontology of an HKB.  For a detailed discussion and the semantics of external atoms, we refer to \cite{eiter2006effective}. What we describe here is a simplification of the much more general hex formalism.

Regular atoms are of the form $p(X_{1}, \ldots, X_{n})$ where $p$ is a predicate symbol of arity $n$ and $X_{1}, \ldots, X_{n}$ are terms, that is, constants or variables. An external atom is of the form $\&g[X_{1},\ldots,X_{n}](Y_{1},\ldots,Y_{m})$
where $g$ is an external predicate name $g$ (which in our case interfaces with the ontology), $X_{1},\ldots,X_{n}$ are input terms and $Y_{1},\ldots,Y_{m}$ are output terms. 

Next, we define the notion of positive rules that may contain external atoms. \begin{definition}\label{rules-hex}
A hex rule $r$ is of the form
$$a \gets b_{1},\ldots,b_{k}. \>\>\>\>\>\> k \geq 0$$
where $a$ is regular atom and $b_{1},\ldots,b_{k}$ are regular or external atoms. We refer to $a$ as the head of $r$, denoted as $H(r)$, while the conjunction $b_{1},...,b_{k}$ is called the body of $r$.
\end{definition}
We call $r$ ordinary if it does not contain external atoms. A program $\mathcal{P}$ containing only ordinary rules is called a positive program, otherwise a hex program. A hex program is a finite set of rules.

The semantics of hex programs generalizes the answer set semantics. The Herbrand base of $\mathcal{P}$, denoted $HB_{\mathcal{P}}$, is the set of all possible ground versions of atoms and external atoms occurring in $\mathcal{P}$ (obtained by replacing variables with constants). Note that constants are not just those in the standard Herbrand universe (those occuring in $\mathcal{P}$) but also those created by external atoms,  which in our case will be IRIs from \Mo. Let the grounding of a rule $r$ be $grd(r)$ and the grounding of program $\mathcal{P}$ be $grd(\mathcal{P}) = \bigcup_{r\in\mathcal{P}} grd(r)$. An interpretation relative to $\mathcal{P}$ is any subset $I \subseteq HB_{\mathcal{P}}$ containing only regular atoms. We write $I \models a$ iff $a \in I$. With every external predicate name $\&g \in G$ we associate an $(n+m+1)$-ary Boolean function $f_{\&g}$ (called oracle function) assigning each tuple $(I,  x_1,\ldots, x_n, y_1 \ldots, y_m)$ either 0 or 1, where $I$ is an interpretation and $x_i, y_j$ are constants. We say that $I \models \&g[x_1,\ldots, x_n](y_1, \ldots, y_m)$ iff $f_{\&g}(I, x_1 \ldots, x_n, y_1,\ldots, y_m) = 1$. For a ground rule $r$,  $I \models B(r)$ iff $I \models a$ for all $a \in B(r)$ and  $I \models r$ iff $I \models H(r)$ whenever $I \models B(r)$. We say that $I$ is a model of $\mathcal{P}$, denoted $I \models \mathcal{P}$, iff $I \models r$ for all $r \in grd(\mathcal{P})$.  The \emph{FLP-reduct} of $\mathcal{P}$ w.r.t $I$, denoted as $f\mathcal{P}^{I}$, is the set of all $r \in grd(\mathcal{P})$ such that $I \models B(r)$. An interpretation $I$ is an answer set of $\mathcal{P}$ iff $I$ is a minimal model of $f\mathcal{P}^{I}$. By $AS(\mathcal{P})$ we denote the set of all answer sets of $\mathcal{P}$. If $\Kcal=\pair{\Mo}{\Pcal}$, then we write $AS(\Kcal) = AS(\Pcal)$ --- note that $\Mo$ is implicitly involved via the external atoms in $\mathcal{P}$. In this paper, $AS(\mathcal{K})$ will always contain exactly one answer set, so we will abuse notation and write $AS(\mathcal{K})$ to denote this unique answer set.

We will also need the notion of query answers of HKBs that contain rules defining a dedicated query predicate $q$. Given a hybrid knowledge base $\mathcal{K}$ and a query predicate $q$, let $ANS(q,\mathcal{K})$ denote the set $\{\langle x_1, \ldots, x_n \rangle \mid q(x_1, \ldots, x_n) \in AS(\mathcal{K})\}$.

\section{Query Answering Using MSER}
We consider SPARQL queries, a W3C standard for querying ontologies. While SPARQL query results can in general either be result sets or RDF graphs, we have restricted ourselves to simple \textbf{SELECT} queries, so it is sufficient for our purposes to denote results by set of tuples. For example, consider the following SPARQL query:\\
\\SELECT$\>?x\> ?y\> ?z\>$WHERE $\>\{ \\
    \indent  \indent  \indent ?x \> rdf\!:\!type \> ?y .\\
    \indent  \indent  \indent ?y \> rdf\!s\!:\!SubClassOf \> ?z \\
     \indent \indent \}$\\
\\This query will retrieve all triples $\langle x,y,z\rangle$, where $x$ is a member of class $y$ that is a subclass of $z$. In general, there will be several variables and there can be multiple matches, so the answers will be sets of tuples of IRIs.

Now, we recall query answering under the Meta-modeling Semantics Entailment Regime (MSER) from \cite{cima2017sparql}. This technique reduces  SPARQL query answering over OWL 2 QL ontologies to Datalog query answering. The main idea of this approach is to define (i) a translation function $\tau$ mapping OWL 2 QL axioms to Datalog facts and (ii)  a fixed Datalog rule base $\mathcal{R}^{ql}$ that captures inferences in OWL 2 QL reasoning.  

The reduction employs a number of predicates, which are used to encode the basic axioms available in OWL 2 QL. This includes both axioms that are explicitly represented in the ontology (added to the Datalog program as facts via $\tau$) and axioms that logically follow. In a sense, this representation is closer to a meta-programming representation than other Datalog embeddings that translate each axiom to a rule.

The function $\tau$ transforms an OWL 2 QL assertion $\alpha$ to a fact. For a given ontology $\mathcal{O}$, we will denote the set of facts obtained by applying $\tau$ to all of its axioms as $\tau(\Mo)$; it will be composed of two portions $\tau(\mathcal{T})$ and $\tau(\mathcal{A})$, as indicated in Table~\ref{tau-function}.\footnote{Note that there are no variables in $\tau(\mathcal{T})$ and $\tau(\mathcal{A})$.}
\begin{table}[!ht]
\centering
\caption{$\tau$ Function}
\label{tau-function}
\begin{tabular}{|c|l|l|ll|}
\noalign{\hrule height 0.5pt}
$\tau(\Mo)$                               & \multicolumn{1}{c|}{$\alpha$}                  & \multicolumn{1}{c|}{$\tau$($\alpha$)} & \multicolumn{1}{c|}{$\alpha$}                                         & \multicolumn{1}{c|}{$\tau$($\alpha$)} \\ \noalign{\hrule height 0.5pt}
\multirow{12}{*}{$\tau(\mathcal{T})$} & c1 $\sqsubseteq$ c2                            & isacCC(c1, c2)                        & \multicolumn{1}{l|}{r1 $\sqsubseteq \neg$ r2}                         & disjrRR(r1,r2)                        \\ \cline{2-5} 
                                                               & c1 $\sqsubseteq \exists$r2$^-$.c2              & isacCI(c1,r2,c2)                    & \multicolumn{1}{l|}{c1 $\sqsubseteq \neg$ c2}                         & disjcCC(c1,c2)                        \\ \cline{2-5} 
                                                               & $\exists$r1 $\sqsubseteq \exists$r2.c2         & isacRR(r1,r2,c2)                       & \multicolumn{1}{l|}{c1 $\sqsubseteq \neg \exists$r2$^-$}              & disjcCI(c1,r2)                        \\ \cline{2-5} 
                                                               & $\exists$r1$^- \sqsubseteq$ c2                 & isacIC(r1,c2)                         & \multicolumn{1}{l|}{$\exists$r1$\sqsubseteq \neg$ c2}                 & disjcRC(r1,c2)                        \\ \cline{2-5} 
                                                               & $\exists$r1$^-$ $\sqsubseteq \exists$r2.c2     & isacIR(r1,r2,c2)                      & \multicolumn{1}{l|}{$\exists$r$_1$ $\sqsubseteq \neg \exists$r2}      & disjcRR(r1,r2)                        \\ \cline{2-5} 
                                                               & $\exists$r1$^-$ $\sqsubseteq \exists$r2$^-$.c2 & isacII(r1,r2,c2)                      & \multicolumn{1}{l|}{$\exists$r1 $\sqsubseteq \neg \exists$r2$^-$}     & disjcRI(r1,r2)                        \\ \cline{2-5} 
                                                               & r1 $\sqsubseteq$ r2                            & isarRR(r1,r2)                         & \multicolumn{1}{l|}{$\exists$r1$^- \sqsubseteq \neg$ c2}              & disjcIC(r1,c2)                        \\ \cline{2-5} 
                                                               & r1 $\sqsubseteq$ r2$^-$                        & isarRI(r1,r2)                         & \multicolumn{1}{l|}{$\exists$r1$^-$ $\sqsubseteq \neg \exists$r2}     & disjcIR(r1,r2)                        \\ \cline{2-5} 
                                                               & c1 $\sqsubseteq \exists$r2.c2                  & isacCR(c1,r2,c2)                      & \multicolumn{1}{l|}{$\exists$r1$^-$ $\sqsubseteq \neg \exists$r2$^-$} & disjcII(r1,r2)                        \\ \cline{2-5} 
                                                               & $\exists$r1$\sqsubseteq$ c2                    & isacRC(r1,c2)                         & \multicolumn{1}{l|}{r1 $\sqsubseteq \neg$ r2$^-$}                     & disjrRI(r1,r2)                        \\ \cline{2-5} 
                                                               & $\exists$r1 $\sqsubseteq \exists$r2$^-$.c2     & isacRI(r1,r2,c2)                      & \multicolumn{1}{l|}{irref(r)}                                         & irrefl(r)                             \\ \cline{2-5} 
                                                               & refl(r)                                        & refl(r)                               &                                                                       &                                       \\ \noalign{\hrule height 0.5pt}
\multirow{2}{*}{$\tau(\mathcal{A})$}  & c(x)                                           & instc(c,x)                            & \multicolumn{1}{l|}{x $\neq$ y}                                       & diff(x,y)                             \\ \cline{2-5} 
                                                               & r(x, y)                                        & instr(r,x,y)                          &                                                                       &                                       \\ \noalign{\hrule height 0.5pt}
\end{tabular}
\end{table}

The fixed program $\mathcal{R}^{ql}$ can be viewed as an encoding of axiom saturation in OWL 2 QL. The full set of rules provided by authors of \cite{cima2017sparql} are reported in the online repository of \cite{qureshi2021evaluation}. We will consider one rule to illustrate the underlying ideas:

\begin{quote}
   \centering isacCR(C1,R2,C2) $\gets$ isacCC(C1,C3), isacCR(C3,R2,C2).
\end{quote}

The above rule encodes the following inference rule:
\begin{quote}
   \centering $\mathcal{O} \models$ C1 $\sqsubseteq$ C3, $\mathcal{O} \models$ C3 $\sqsubseteq \exists$R2.C2 $ \Rightarrow  \mathcal{O} \models$ C1 $\sqsubseteq \exists$R2.C2
\end{quote}

Finally, the translation can be extended in order to transform conjunctive SPARQL queries under MS over OWL 2 QL ontologies into a Datalog query. SPARQL queries will be translated to Datalog rules using a transformation $\tau^{q}$.
$\tau^{q}$ uses $\tau$ to translate the triples inside the body of the SPARQL query $\mathcal{Q}$ and adds a fresh Datalog predicate $q$ in the head to account for projections. In the following we assume $q$ to be the query predicate created in this way.

For example, the translation of the SPARQL query given earlier will be
\begin{quote}
  \centering q(X,Y,Z) $\gets$ instc(X,Y), isacCC(Y,Z).\\
\end{quote}

Given an OWL 2 QL ontology \Mo and a SPARQL query $\mathcal{Q}$, let $ANS(\mathcal{Q}, \Mo)$ denote the answers to $\mathcal{Q}$ over \Mo under MSER, that is, a set of tuples of IRIs. In the example above, the answers will be a set of triples.

\section{MSER Query Answering via Hybrid Knowledge Bases}

We propose four variants for answering MSER queries by means of Hybrid Knowledge Bases. We first describe the general approach and then define each of the four variants.

\subsection{General Architecture}

The general architecture is outlined in Figure~\ref{fig1}. In all cases, the inputs are an OWL 2 QL ontology \Mo and a SPARQL query $\mathcal{Q}$. We then differentiate between $\mathbf{Ontology Functions}$ and $\mathbf{Query Functions}$. The $\mathbf{Ontology Functions}$ achieves two basic tasks: first, the ontology is split into two partitions \Oprime and \Opprime, then $\tau(\Opprime)$ is produced. 

The $\mathbf{Query Functions}$ work mainly on the query. First, a set $\mathcal{N}$ of IRIs is determined for creating \emph{Interface Rules} ($\mathrm{IR}$, simple hex rules), denoted as $\pi(\mathcal{N})$ for importing the extensions of relevant classes and properties from $\Oprime$. In the simplest case, $\mathcal{N}$, consist of all IRIs in $\Oprime$, but we also consider isolating those IRIs that are relevant to the query by means of Logic-based Module Extraction (LME) as defined in \cite{jimenez2008safe}.
 Then, $\tau^{q}$ translates $\mathcal{Q}$ into a Datalog query $\tau^{q}(\mathcal{Q})$. Finally, the created hex program components are united (plus the fixed inference rules), yielding the rule part $\mathcal{P} = \mathcal{R}^{ql} \cup \pi(\mathcal{N}) \cup \tau(\Opprime) \cup \tau^{q}(\mathcal{Q})$, which together with \Oprime forms the HKB $\Kcal = \pair{\Oprime}{\mathcal{P}}$, for which we then determine  $ANS(q,\Kcal)$, where $q$ is the query predicate introduced by $\tau^{q}(\mathcal{Q})$.
\begin{figure}[!ht]{
 \Large
 \caption{The Overall Architecture of Hybrid-Framework} \label{fig1}
 \begin{center}
 \scalebox{0.4}{\begin{tikzpicture}
\draw [rounded corners=0.2cm,inner sep=0pt,dashed,fill=white!20] (-6,1) rectangle ++(11.5,2.5) node at (-7.7,2.25) [text=black!70,  text width=3cm] {\Large \textsc{Input}};
 \draw[->,ultra thick,green!50, fill opacity=0.4 ] (-2.3,1.5) -- (-2.3,-0.5);
\draw [draw=green!40,rounded corners=0.3cm,fill=green!40, fill opacity=0.5] (-5.6,1.5) rectangle ++(5,1.5);
  \node[text width=8cm,align=center] at (-3.2,2.25)  {$\mathrm{Ontology} (\Mo)$};
\draw [draw=red!40,rounded corners=0.3cm,fill=red!40, fill opacity=0.5] (0.1,1.5) rectangle ++(5,1.5);
  \node[text width=8cm,align=center] at (2.5,2.25)  {$Query (\mathcal{Q})$ };
\draw [rounded corners=0.2cm,inner sep=0pt,dashed,fill=white!20] (-6,-3) rectangle ++(8.4,2.5) node at (-8,-1.75) [text=black!70, text width=2.5cm] {\Large \textsc{Ontology\\Functions}};
\draw [rounded corners=0.2cm,inner sep=0pt,dashed,fill=white!20] (-3.1,-7) rectangle ++(8.3,2.5) node at (-8,-5.5) [text=black!70, text width=2.5cm] {\Large \textsc{Query\\Functions}};
\draw [rounded corners=0.2cm,inner sep=0pt,dashed,fill=white!20] (-6,-10.3) rectangle ++(10,1.5) node at (-8,-9.5) [text=black!70, text width=2.5cm] {\Large \textsc{Evaluate\\Task}};
  \node[text width=8.5cm,align=center] at (-1.15,-9.55)  {$ANS(q, \langle \mathcal{O'},\mathcal{R}^{ql} \cup \pi(\mathcal{N}) \cup \tau(\Opprime) \cup \tau^{q}(\mathcal{Q})\rangle)$ };
  \draw[->,ultra thick,aauBlue!50, fill opacity=0.4 ] (-4.3,-3) -- (-4.3,-8.8);
  \draw[->,ultra thick,red!30!white, fill opacity=0.4 ] (3.5,1.0) -- (3.5,-4.5);
  \draw[->,ultra thick,aauBlue2!50, fill opacity=0.4 ] (-1,-7) -- (-1,-8.8);
  \draw [rounded corners=0.2cm,fill=aauBlue!20] (-5.7,-2.7) rectangle ++(3.6,2) node [midway,text width=5cm,align=center,font=\normalsize] 
  {Partition \Mo:\\$\Oprime$\\$\Opprime$};
  \draw [rounded corners=0.2cm,fill=aauBlue!20] (-1.5,-2.7) rectangle ++(3.5,2) node [midway,text width=5cm,align=center,font=\normalsize] {Translate\mbox{} \\$\tau(\Opprime)$};
  \draw [rounded corners=0.2cm,fill=aauBlue!20] (-2.8,-6.7) rectangle ++(3.6,2) node [midway,text width=5cm,align=center,font=\normalsize] 
  {Create \emph{Interface Rules}\\$\pi(\mathcal{N})$ };
  \draw [rounded corners=0.2cm,fill=aauBlue!20] (1.3,-6.7) rectangle ++(3.5,2) node [midway,text width=5cm,align=center,font=\normalsize] {Translate query \\ $\tau^{q}(\mathcal{Q})$};
\end{tikzpicture}
}
   
 \end{center}}
\end{figure}
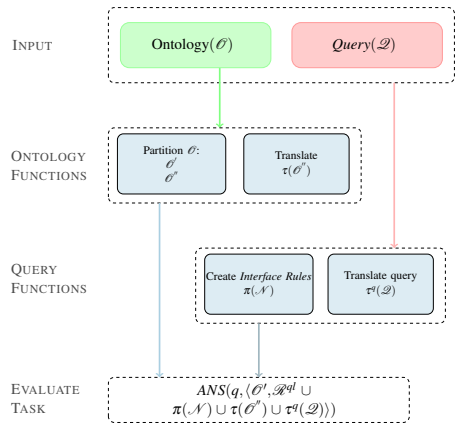

\subsection{Basic Notions}

Before defining the specific variations of our approach, we first define some auxiliary notions. The first definition identifies meta-elements.

\begin{definition}\label{metaElement}
Given an Ontology \Mo, IRIs in  $ ( V_c \cup V_p ) \cap V_i$ are meta-elements, i.e., IRIs that occur both as individuals and classes or object properties.
\end{definition}

In our example, \textit{GoldenEagle} is a meta-element. Meta-elements form the basis of our main notion, clashing axioms.

\begin{definition}\label{def:clashingAxioms}
Clashing Axioms in \Mo are axioms that contain meta-elements, denoted as $\mathrm{CA}(\Mo)$. To denote clashing and non-clashing parts in TBox ($\mathcal{T}$) and ABox ($\mathcal{A}$), we write  $\mathcal{A}^N = \mathcal{A}\setminus \mathrm{CA}(\Mo)$ as non-clashing ABox, $\mathcal{A}^C = \mathrm{CA}(\Mo) \cap \mathcal{A}$ as clashing ABox; and likewise $\mathcal{T}^N = \mathcal{T}\setminus \mathrm{CA}(\Mo)$ as non-clashing TBox and $\mathcal{T}^C = \mathrm{CA}(\Mo) \cap \mathcal{T}$ as clashing TBox.
\end{definition}

The clashing axiom notion allows for splitting \Mo into two parts and generate \Oprime without clashing axioms. 

We would also like to distinguish between standard queries and meta-queries.
A meta-query is an expression consisting of meta-predicates $p$ and meta-variables $v$, where $p$ can have other predicates as their arguments and $v$ can appear in predicate positions.
The simplest form of meta-query is an expression where variables appear in class or property positions also known as \emph{second-order queries}. More interesting forms of meta-queries allow one to extract complex patterns from the ontology, by allowing variables to appear simultaneously in individual object and class or property positions. We will refer to non-meta-queries as standard queries. Moving towards \emph{Interface Rules}, we first define signatures of queries, ontologies, and axioms.

\begin{definition}\label{def:Signatures}
A \emph{signature} $\Sig{\Qcal}$ of a SPARQL query $\Qcal$ is the set of \emph{IRIs} occurring in $\mathcal{Q}$. If no IRIs occur in $\mathcal{Q}$, we define $\mathbf{S}(\mathcal{Q})$ to be the signature of \Mo. Let $\mathbf{S}(\Mo)$ (or $\mathbf{S}(\alpha)$) be the set of atomic classes, atomic roles and individuals that occur in \Mo (or in axiom $\alpha$).
\end{definition}

As hinted earlier, we can use $\mathbf{S}(\Oprime)$ for creating interface rules ($\Oprime$ being the ontology part in the HKB), or use $\mathbf{S}(\mathcal{Q})$ for module extraction  via \emph{LME} as defined in \cite{jimenez2008safe} for singling out the identifiers relevant to the query, to be imported from the ontology via interface rules. We will denote this signature as $\Sig{\LME(\Sig{\Qcal},\Oprime)}$.

We next define the \textit{Interface Rules} for a set of IRIs $\mathcal{N}$.

\begin{definition}\label{def:pi(N)}
  For a set a of IRIs $\mathcal{N}$, let $\pi(\mathcal{N})$ denote the hex program containing a rule
  \begin{quote}
  \centering $instc(C,X)\>\gets\>\>\&g[C](X).$\\
  \end{quote}
  for each class identifier $C \in \mathcal{N}$, and a rule
    \begin{quote}
  \centering $instr(R,X,Y)\>\gets\>\>\&g[R](X,Y).$
    \end{quote}
    for each property identifier $R \in \mathcal{N}$.
    Here $\&g$ is a shorthand for the external atom that imports the extension of classes or properties from the ontology $\Oprime$ of our framework.\footnote{Note that $C$ and $R$ above are not variables, but IRIs.}
\end{definition}


\subsection{Variants}

Now we define the four variants for the ontology functions, and two for the query functions. Since for one ontology function $\Oprime$ is empty, the two query functions have the same effect, and we therefore arrive at seven different variants for creating the hybrid knowledge bases (HKB). 

The difference in the ontology functions is which axioms of $\Mo=\pair{\Acal}{\Tcal}$ stay in $\Oprime$ and which are in \Opprime, the latter of which is translated to Datalog. We use a simple naming scheme,  indicating these two components:

\begin{description}
\item[\AT:]  $\Oprime=\Acal$, $\Opprime=\Tcal$.
\item[\NATCAT:] $\Oprime=\pair{\AN}{\Tcal}$, $\Opprime=\pair{\AC}{\Tcal}$.
\item[\NATCACT:] $\Oprime=\pair{\AN}{\Tcal}$, $\Opprime=\pair{\AC}{\TC}$.
\item[\EAT:] $\Oprime=\emptyset$, $\Opprime=\Mo=\pair{\Acal}{\Tcal}$.
\end{description}

\EAT{} serves as a baseline, as it boils down to the Datalog encoding of \cite{cima2017sparql}. 

\begin{definition}
Given $\Mo=\pair{\Acal}{\Tcal}$, let the \AT{} HKB be $\KAT{\Mo} = \pair{\Acal}{\Rql\cup\ontotrans{\Tcal}}$;  the \NATCAT{} HKB be $\KNATCAT{\Mo} = \pair{\pair{\AN}{\Tcal}}{\Rql\cup\ontotrans{\pair{\AC}{\Tcal}}}$; the \NATCACT{} HKB be $\KNATCACT{\Mo} = \pair{\pair{\AN}{\Tcal}}{\Rql\cup\ontotrans{\pair{\AC}{\TC}}}$; the \EAT{} HKB be $\KEAT{\Mo} = \pair{\emptyset}{\Rql\cup\ontotrans{\Mo}}$.

\end{definition}

Next we turn to the query functions. As hinted at earlier, we will consider two versions, which differ in the Interface Rules they create. Both create query rules $\querytrans{\Qcal}$ for the given query, but one (\All) will create interface rules for all classes and properties in the ontology part of the HKB, while the other (\Mod) will extract the portion of the ontology relevant to query using \LME{} and create Interface Rules only for classes and properties in this module.

For notation, we will overload the $\cup$ operator for HKBs, so we let $\pair{\Mo}{\Pcal} \cup \pair{\Oprime}{\Pcal'} = \pair{\Mo\cup\Oprime}{\Pcal\cup\Pcal'}$ and we also let $\pair{\Mo}{\Pcal} \cup \Pcal' = \pair{\Mo}{\Pcal\cup\Pcal'}$ for ontologies $\Mo$ and $\Oprime$ and hex programs $\Pcal$ and $\Pcal'$.

\begin{definition}\label{def:KAll}
Given an HKB $\pair{\Mo}{\Pcal}$ and query $\Qcal$, let the $\All$ HKB be defined as $\KAll{\pair{\Mo}{\Pcal},\Qcal} = \pair{\Mo}{\Pcal\cup\querytrans{\Qcal}\cup\IR{\Sig{\Mo}}}$.
\end{definition}

\begin{definition}\label{def:KMod}
Given an HKB $\pair{\Mo}{\Pcal}$ and query $\Qcal$, let the $\Mod$ HKB be 
$\KMod{\pair{\Mo}{\Pcal},\Qcal} = \pair{\Mo}{\Pcal\cup\querytrans{\Qcal}\cup\IR{\Sig{\LME(\Sig{\Qcal},\Mo)}}}$.
\end{definition}

We will combine ontology functions and query functions, and instead of $\K{}{\beta}(\K{\alpha}{}(\Mo),\Qcal)$ we will write $\K{\alpha}{\beta}(\Mo,\Qcal)$. We thus get eight combinations, but we will not use $\K{\EAT}{\Mod}$, as it unnecessarily introduces Interface Rules. Also note that $\K{\EAT}{\All}(\Mo,\Qcal)$ does not contain any Interface Rules, because the ontology part of $\KEAT{\Mo}$ is empty.

We will next show the correctness of the transformations. We start with the simplest case.

\begin{proposition}\label{propbase}
Let \Mo be a consistent OWL 2 QL ontology and  $\Qcal$ a conjunctive SPARQL query. Then, $ANS(\mathcal{Q},\Mo) = ANS(q,\K{\EAT{}}{\All}(\Mo,\Qcal))$, where $q$ is the query predicate introduced by $\tau^{q}(\mathcal{Q})$.
\end{proposition}
\begin{proof}
  In \cite{cima2017sparql} it was shown that $ANS(\mathcal{Q},\Mo) = P^{q}(\tau(\Mo)) = \{ \langle x_1,\ldots,x_n \rangle \mid q(x_1,\ldots,x_n) \in MM(\Rql\cup\ontotrans{\Mo}\cup \querytrans{\Qcal}) \}$.

  Since $MM(P) = AS(P) = AS(\pair{\emptyset}{P})$ for any Datalog program $P$,  it follows that $ANS(\mathcal{Q},\Mo) = \{ \langle x_1,\ldots,x_n \rangle \mid q(x_1,\ldots,x_n) \in AS(\pair{\emptyset}{\Rql\cup\ontotrans{\Mo}\cup \querytrans{\Qcal}}) \}$.

  Per definition, we get $\K{\EAT}{\All}(\Mo,\Qcal) = \KAll{\K{\EAT}{}(\Mo),\Mo,\Qcal} = \KAll{\pair{\emptyset}{\Rql\cup\ontotrans{\Mo}},\Mo,\Qcal} = \pair{\emptyset}{\Rql\cup\ontotrans{\Mo}\cup \querytrans{\Qcal} \cup \IR{\Sig \Mo}}$, therefore $ANS(q,\K{\EAT{}}{\All}(\Mo,\Qcal)) = \{ \langle x_1,\ldots,x_n \rangle \mid q(x_1,\ldots,x_n) \in AS(\pair{\emptyset}{\Rql\cup\ontotrans{\Mo}\cup \querytrans{\Qcal} \cup \IR{\Sig \Mo}})\}$.

  We now show $AS(\pair{\emptyset}{\Rql\cup\ontotrans{\Mo}\cup \querytrans{\Qcal}}) = AS(\pair{\emptyset}{\Rql\cup\ontotrans{\Mo}\cup \querytrans{\Qcal} \cup \IR{\Sig \Mo}})$, which proves the proposition. Indeed, for any interpretation $I$ we have that  $I \not\models r$ for each $r\in\IR{\Sig \Mo}$, because the ontology of the hybrid knowledge base is empty.
 Hence $f\pair{\emptyset}{\Rql\cup\ontotrans{\Mo}\cup \querytrans{\Qcal}}^I = f\pair{\emptyset}{\Rql\cup\ontotrans{\Mo}\cup \querytrans{\Qcal} \cup \IR{\Sig \Mo}}^I$ for any interpretation $I$, and the equality of answer sets follows.
\end{proof}

\begin{theorem}\label{theorem1}
  Let \Mo be a consistent OWL 2 QL ontology,  $\Qcal$ a conjunctive SPARQL query, then it holds that $ANS(\mathcal{Q},\Mo) = ANS(q,\K{\alpha}{\All}(\Mo,\Qcal))$,  where $\alpha$ is one of \AT{}, \NATCAT{}, or \NATCACT{} and where $q$ is the query predicate introduced by $\tau^{q}(\mathcal{Q})$. 
  
\end{theorem}

\begin{proof}
  From Proposition~\ref{propbase} we have that $ANS(\mathcal{Q},\Mo) = ANS(q,\K{\EAT{}}{\All}(\Mo,\Qcal))$. We now show that $AS(\K{\EAT{}}{\All}(\Mo,\Qcal)) = AS(\K{\alpha}{\All}(\Mo,\Qcal))$ and $ANS(q,\K{\EAT{}}{\All}(\Mo,\Qcal)) = ANS(q,\K{\alpha}{\All}(\Mo,\Qcal))$ follows.

  First, $\K{\EAT{}}{\All}(\Mo,\Qcal) =  \pair{\emptyset}{\Rql\cup\ontotrans{\Mo}\cup \querytrans{\Qcal} \cup \IR{\Sig \Mo}}$ (for short $E$), and let $\K{\alpha}{\All}(\Mo,\Qcal)) = \pair{\Oprime}{\Rql\cup\ontotrans{\Opprime}\cup \querytrans{\Qcal} \cup \IR{\Sig \Mo}}$ (for short $A$). In all cases, $\Oprime \subseteq \Mo$, $\Opprime \subseteq \Mo$ and $\Oprime \cup \Opprime = \Mo$. Moreover, $\Mo \models \varphi$ ($\varphi$ atomic over $\Sig \Mo$) if and only if $\Opprime \cup \{\psi \mid \Oprime \models \psi, \psi \mbox{\ atomic over\ } \Sig \Mo\} \models \varphi$, let us call this the ontology splitting property.

Now, for any interpretation $I$, $fE^I \neq fA^I$ may hold, but for any interpretation $J$, $J \models fE^I$ if and only if $J\models fA^I$. This is because for each atomic $\varphi$ over $\Sig \Mo$, either $\Oprime \models \varphi$, then there is a rule in $\IR{\Sig \Mo}$ with a true body in $fA^I$ (because of $\Oprime$) and $\ontotrans{\varphi}$ in its head. That rule is satisfied by $J$ iff $\ontotrans{\varphi} \in J$. For $fE^I$, because of the results of \cite{cima2017sparql} there is a rule in $\ontotrans{\Mo}$ with $\ontotrans{\varphi}$ in its head and a true body; also that rule is satisfied by $J$ iff $\ontotrans{\varphi} \in J$. If $\Oprime \not \models \varphi$, then $\Opprime \cup \{\psi \mid \Oprime \models \psi, \psi \mbox{\ atomic over\ } \Sig \Mo\} \models \varphi$. In that case, the same rule with $\ontotrans{\varphi}$ in its head is both in $fA^I$ and $fE^I$.

Since $J \models fE^I$ if and only if $J\models fA^I$, also the minimal models of $fE^I$ and $fA^I$ are the same, and from this $AS(\K{\EAT{}}{\All}(\Mo,\Qcal)) = AS(\K{\alpha}{\All}(\Mo,\Qcal))$ follows.

\end{proof}

Note that the same proof also works for potential other variants that satisfy the ontology splitting property.

\begin{theorem}\label{theorem2}
  Let \Mo be a consistent OWL 2 QL ontology,  $\Qcal$ a conjunctive SPARQL query, then  it holds that $ANS(\mathcal{Q},\Mo) = ANS(q,\K{\alpha}{\Mod}(\Mo,\Qcal))$,  where $\alpha$ is one of \AT{}, \NATCAT{}, or \NATCACT{} and where $q$ is the query predicate introduced by $\tau^{q}(\mathcal{Q})$. 
  
\end{theorem}

\begin{proof}

  Note that $\LME(\Sig{\Qcal},\Mo)$ is a module of $\Mo$ in the sense of \cite{jimenez2008safe}. This implies that for any atomic axiom $\varphi$ over $\Sig{\Qcal}$, $\Mo \models \varphi$ iff $\LME(\Sig{\Qcal},\Mo) \models \varphi$. It follows that $ANS(\mathcal{Q},\Mo) = ANS(\mathcal{Q},\LME(\Sig{\Qcal},\Mo)$. We have $ANS(\mathcal{Q},\LME(\Sig{\Qcal},\Mo) = ANS(q,\K{\alpha}{\All}(\LME(\Sig{\Qcal},\Mo),\Qcal))$ from Theorem~\ref{theorem1}. 
$\K{\alpha}{\Mod}(\Mo,\Qcal)= \pair{\Oprime}{\Rql\cup\ontotrans{\Opprime}\cup \querytrans{\Qcal} \cup \IR{\Sig{\LME(\Sig{\Qcal},\Mo)}}}$ is very similar to $\K{\alpha}{\All}(\LME(\Sig{\Qcal},\Mo),\Qcal))$, 
which expands to 
$\pair{\LME(\Sig{\Qcal},\Mo)'}{\Rql\cup\ontotrans{\LME(\Sig{\Qcal},\Mo)''}\cup \querytrans{\Qcal} \cup \IR{\Sig{\LME(\Sig{\Qcal},\Mo)}}}$.
$\K{\alpha}{\Mod}(\Mo,\Qcal)$ just has the larger underlying ontology $\Mo$. $\Oprime$ may contain more axioms than $\LME(\Sig{\Qcal},\Mo)'$, but since the interface rules $\IR{\Sig{\LME(\Sig{\Qcal},\Mo)}}$ are the same in both HKBs, they have no effect. Also $\ontotrans{\Opprime}$ may contain more rules than $\ontotrans{\LME(\Sig{\Qcal},\Mo)''}$, but none of them is relevant to $q$ by definition. So eventually we get $ANS(q,\K{\alpha}{\All}(\LME(\Sig{\Qcal},\Mo),\Qcal)) = ANS(q,\K{\alpha}{\Mod}(\Mo,\Qcal)$, from which the result follows.

\end{proof}

\section{Evaluation}
In \cite{QureshiF23} we conducted experiments using HEXLite with the OWL-API plugin. While it did show drastic improvements when using one of the hybrid approaches with respect to the baseline \EAT and with using $\Mod$ rather then $\All$, the absolute performance left to be desired. In particular, with the larger ontologies considered, no answer could be obtained even after hours. This contrasts sharply with the findings in \cite{qureshi2021evaluation}, in which the best systems took only seconds to answer queries even on the larger ontologies. The main reasons appeared to be inefficiencies in the OWL-API plugin, paired with a lack of query-oriented computation.

In the meantime we became aware of \textit{DLV2} with \textit{Python external atoms}\footnote{\url{https://dlv.demacs.unical.it/home}}.

The version of DLV2 that we obtained from the developers directly supports the Turtle format of ontologies, and one can use ontology IRIs directly as predicate names. The rules in Definition~\ref{def:pi(N)} can then directly use class and role identifiers:

\begin{definition}\label{def:pi(N)-2}
  For a set a of IRIs $\mathcal{N}$, let $\pi(\mathcal{N})$ denote the \textit{DLV2} program containing a rule
  \begin{quote}
  \centering $instc(C,X)\>\gets\>\>\ C(X).$\\
  \end{quote}
  for each class identifier $C \in \mathcal{N}$, and a rule
    \begin{quote}
  \centering $instr(R,X,Y)\>\gets\>\>\ R(X,Y).$
    \end{quote}
    for each property identifier $R \in \mathcal{N}$.
\end{definition}

For transforming our ontologies to Turtle format, we have used a utility called \textit{ont-converter}\footnote{\url{https://github.com/sszuev/ont-converter}} that automatically transforms the source ontology in different formats (RDF/XML, OWL/XML, N3, etc).

The experimental setting is the same as in \cite{QureshiF23}: we conducted two sets of experiments on the widely used Lehigh University Benchmark (LUBM) dataset and on the Making Open Data Effectively USable (MODEUS) Ontologies\footnote{http://www.modeus.uniroma1.it/modeus/node/6}. We only use the query function $\Mod$, as it was evident in \cite{QureshiF23} that $\All$ has no advantage over $\Mod$.

The \textbf{LUBM} datasets describe a university domain with information like departments, courses, students, and faculty. This dataset comes with 14 queries with different characteristics (low selectivity vs high selectivity, implicit relationships vs explicit relationships, small input vs large input, etc.). We have also considered the meta-queries \textit{mq1}, \textit{mq4}, \textit{mq5}, and \textit{mq10} from \cite{kontchakov2014answering} as they contain variables in-property positions and are long conjunctive queries.  We have also considered two special-case queries \textit{sq1} and \textit{sq2} from \cite{cima2017sparql} to exercise the MSER features and identify the new challenges introduced by the additional expressivity over the ABox queries.  Basically, in special-case queries, we check the impact of \texttt{DISJOINTWITH} and meta-classes in a query. For this, like in \cite{cima2017sparql}, we have introduced a new class named \textit{TypeOfProfessor} and make \textit{FullProfessor}, \textit{AssociateProfessor} and \textit{AssistantProfessor} instances of this new class and also define \textit{FullProfessor}, \textit{AssociateProfessor} and \textit{AssistantProfessor} to be disjoint from each other. Then, in \textit{sq1} we are asking for all those $y$ and $z$, where $y$ is a professor, $z$ is a type of professor and $y$ is an instance of $z$. In \textit{sq2}, we have asked for different pairs of professors. 

The \textbf{MODEUS} ontologies describe the \textit{Italian Public Debt} domain with information like financial liability or financial assets to any given contracts \cite{lenzerini2020metaquerying}. It comes with 8 queries. These queries are pure meta-queries as they span over several levels of the knowledge base. MODEUS ontologies are meta-modeling ontologies with meta-classes and meta-properties.

We have done the experiments on a Linux batch server, running Ubuntu 20.04.3 LTS (GNU/Linux 5.4.0-88-generic x86\_64) on one AMD EPYC 7601 (32-Core CPU), 2.2GHz, Turbo max. 3.2GHz. The machine is equipped with 512GB RAM and a 4TB hard disk. Java applications used OpenJDK 11.0.11 with a maximum heap size of 25GB. During the course of the evaluation of the proposed variants we have used the time resource limitation as the benchmark setting on our data sets to examine the behavior of different variants. If not otherwise indicated, in both experiments, each benchmark had 3600 minutes (excluding the $\mathcal{K}$ generation time). For simplicity, we have not included queries that contain data properties in our experiments.
We also have included the generation time of the hybrid knowledge base $\mathcal{K}$ including the loading of ontology and query, $\tau$ translation, module extraction, generating IR and translating queries.
All material of experiments and results are available at \url{https://doi.org/10.5281/zenodo.13358935}.

\begin{figure}
   \begin{minipage}{0.45\textwidth}
     \centering
          \includegraphics[width=\textwidth]{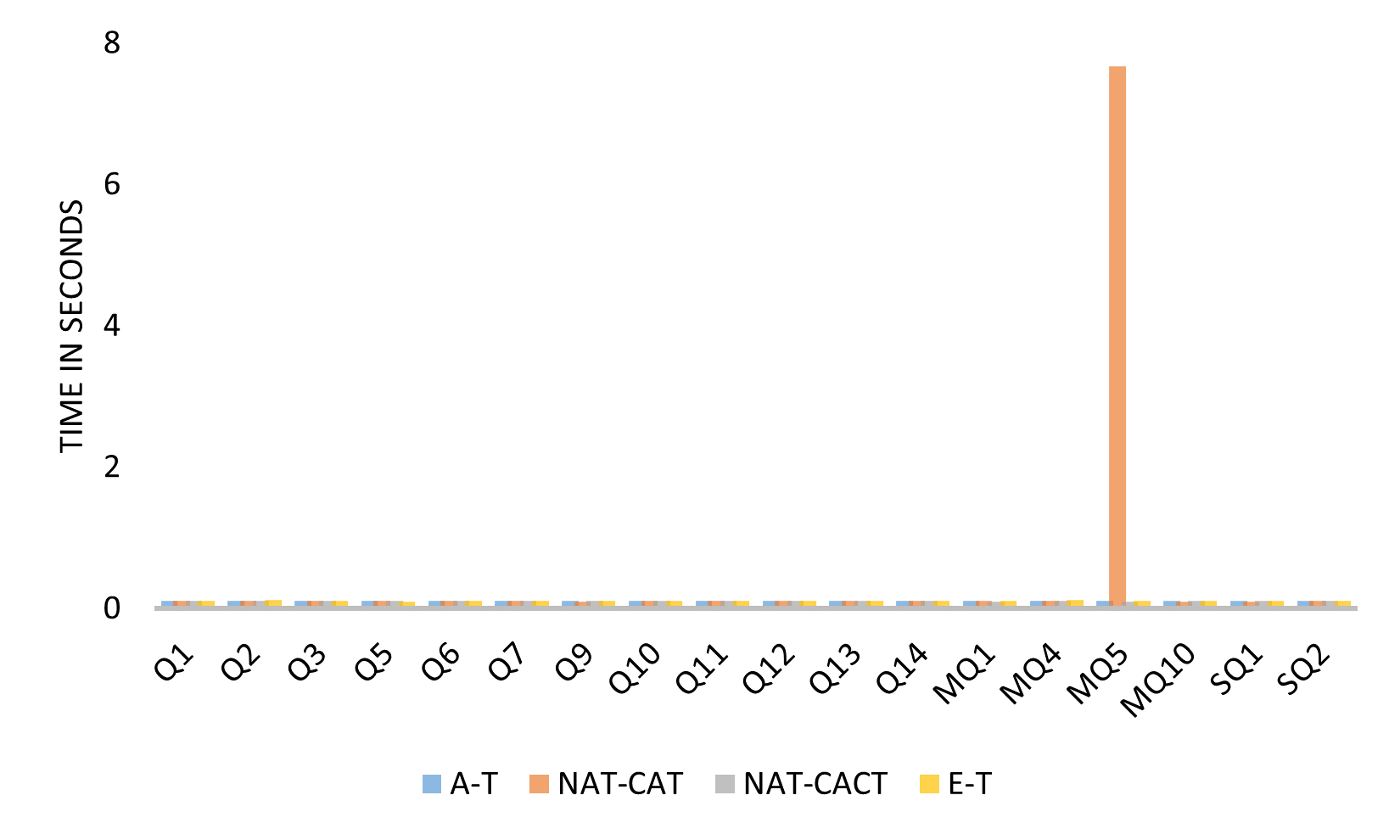}
     \caption{LUBM(1) experiments with standard and meta queries}
  \label{fig:DLV2L1wSMQ}
   \end{minipage}
   \begin{minipage}{0.45\textwidth}
     \centering
     \includegraphics[width=\textwidth]{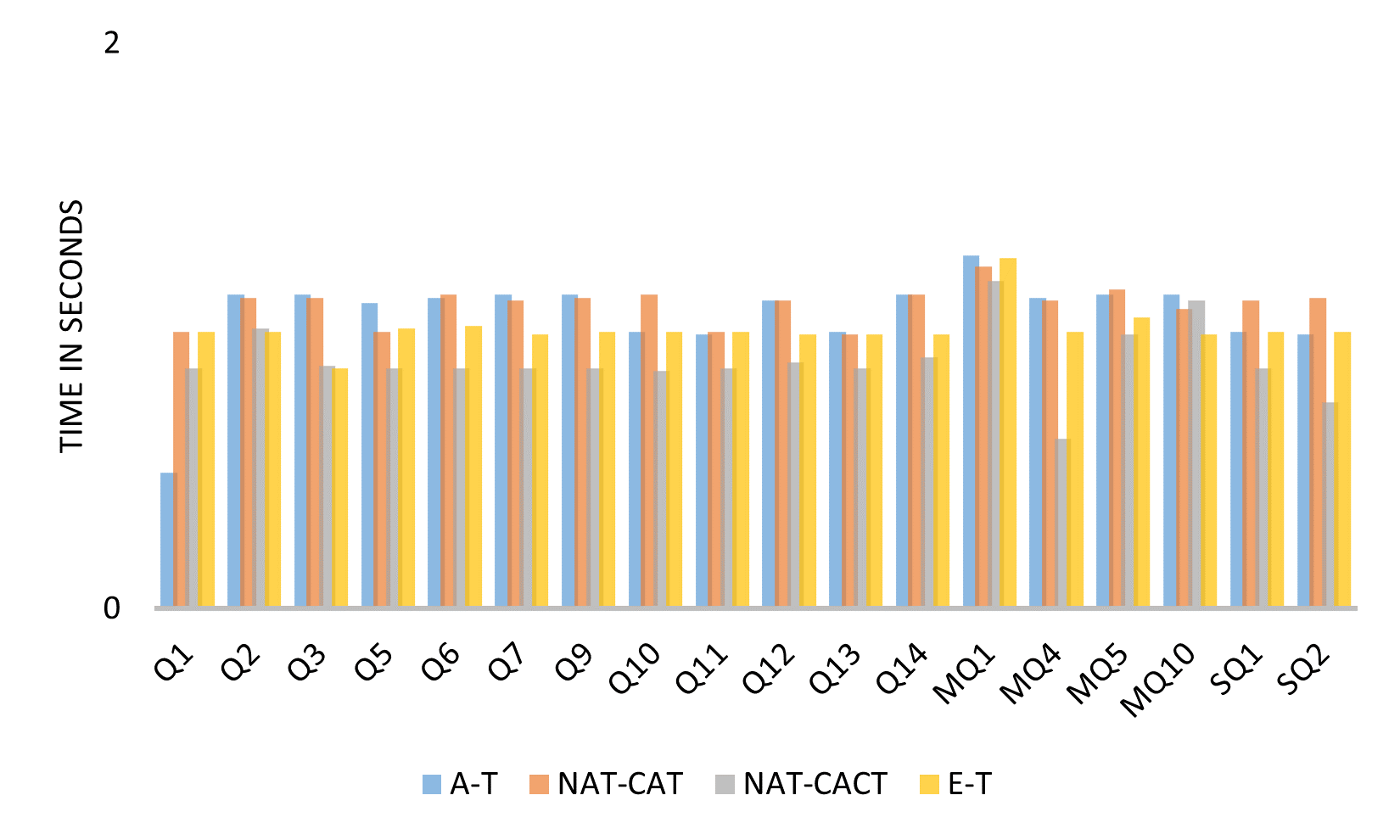}
     \caption{LUBM(9) experiments with Standard and Meta Queries}
  \label{fig:DLV2L9wSMQ}
   \end{minipage}
\end{figure}

In Figure~\ref{fig:DLV2L1wSMQ} and \ref{fig:DLV2L9wSMQ}, it can be seen that \textit{DLV2} shows regular performance across all datasets and all variants of HKB with a slight increase in time depending on the size of the dataset. There is one outlier, meta-query MQ5 on LUBM(1) with \NATCAT, which we were not expecting and might be a measurement error. In any case, this a massive improvement over the performance with HEXLite, where some of these queries required thousands of seconds to evaluate.

\begin{figure}
   \begin{minipage}{0.45\textwidth}
     \centering
     \includegraphics[width=\textwidth]{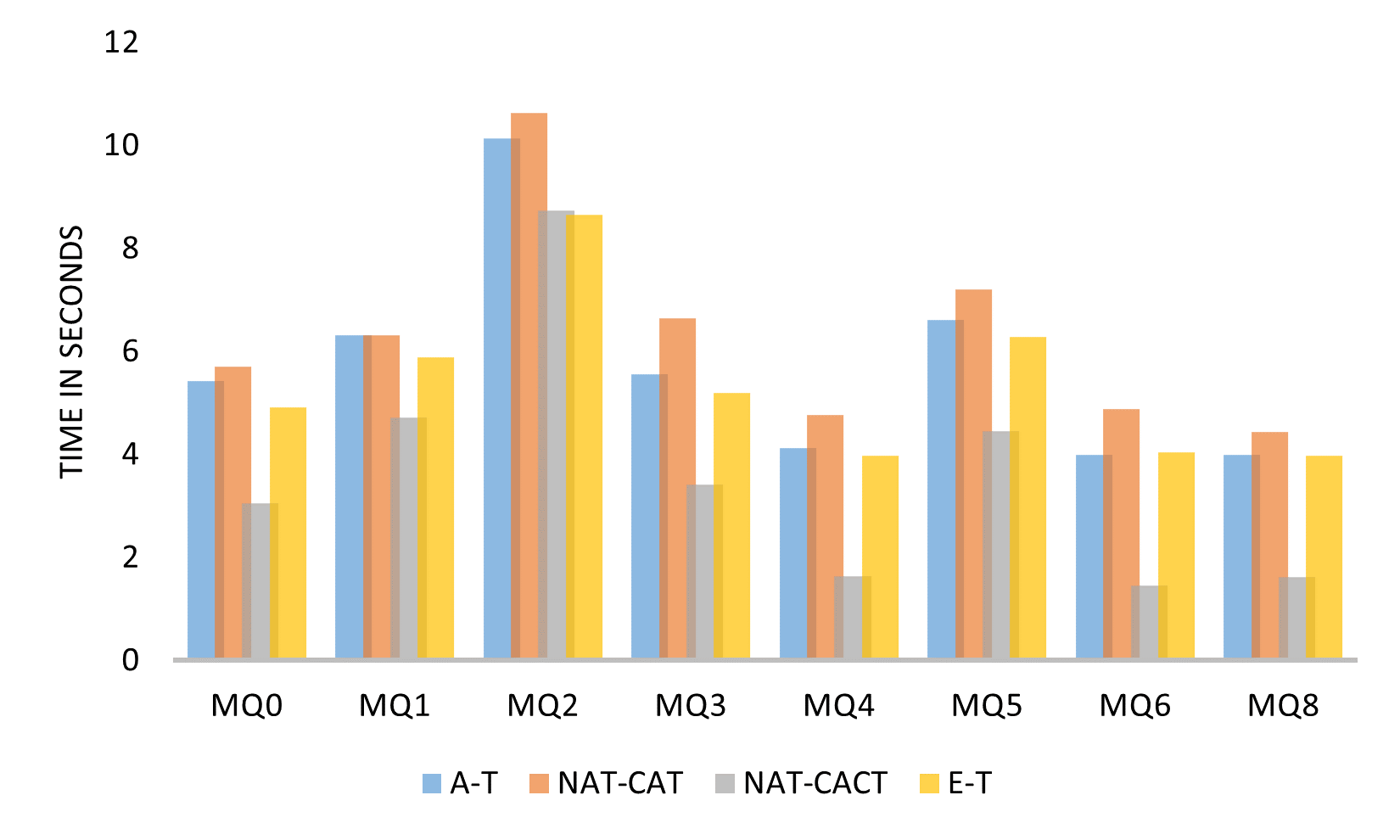}
     \caption{MODEUS(00) with Meta Queries}
  \label{fig:DLV2M0}
   \end{minipage}\hfill
   \begin{minipage}{0.45\textwidth}
     \centering
     \includegraphics[width=\textwidth]{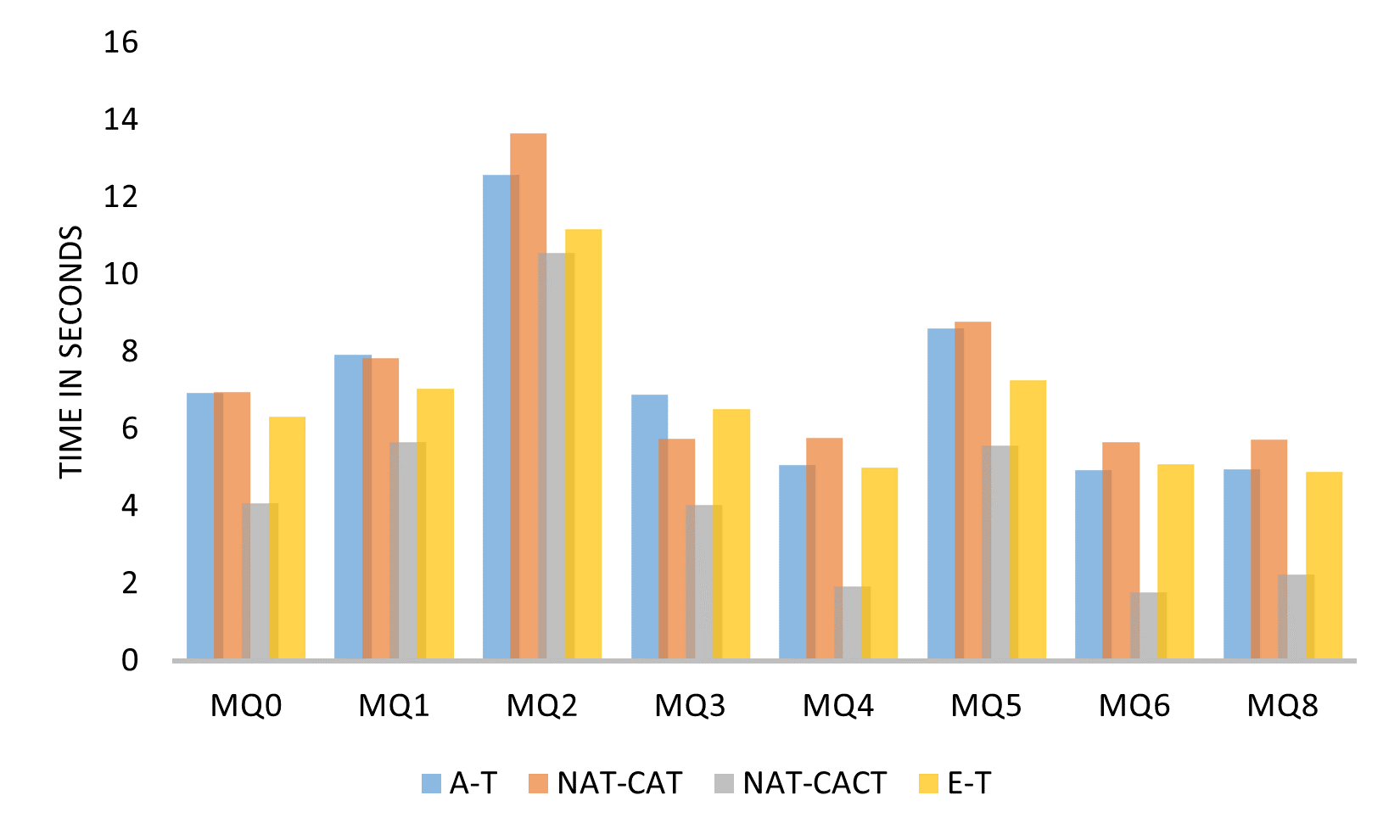}
     \caption{MODEUS(01) with Meta Queries}
  \label{fig:DLV2M1}
   \end{minipage}
\end{figure}

\begin{figure}
   \begin{minipage}{0.45\textwidth}
     \centering
     \includegraphics[width=\textwidth]{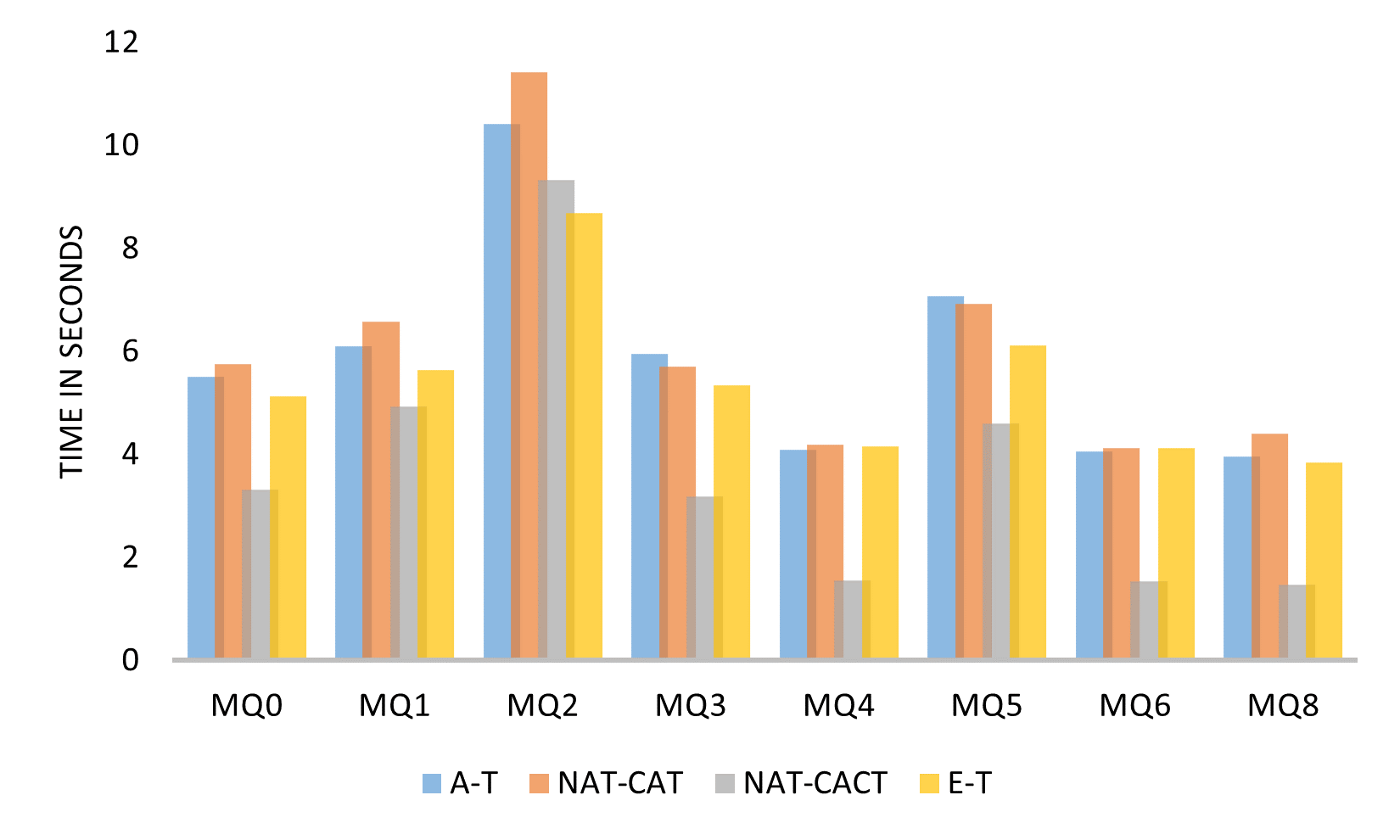}
     \caption{MODEUS(02) with Meta Queries}
  \label{fig:DLV2M2}
   \end{minipage}
   \begin{minipage}{0.45\textwidth}
     \centering
     \includegraphics[width=\textwidth]{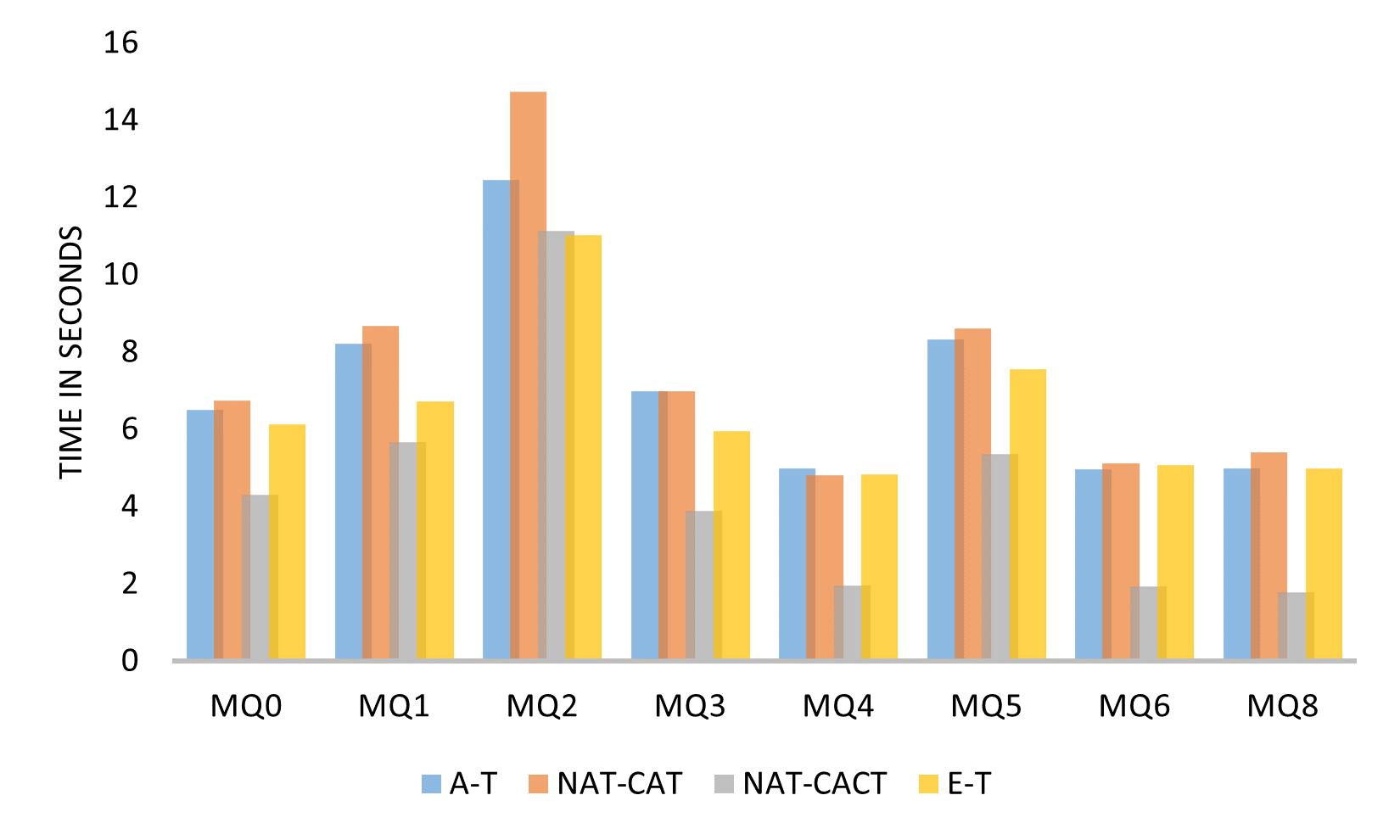}
     \caption{MODEUS(03) with Meta Queries}
  \label{fig:DLV2M3}
   \end{minipage}
\end{figure}

In Figures~\ref{fig:DLV2M0} to \ref{fig:DLV2M3} the performance on MODEUS queries is reported. All the variants show consistent performance; however, the behaviour of the \NATCACT variant seems to be usually the best. These results are very satisfactory with respect to the results observed with HEXLite, where none of these queries were answered even after a few hours of runtime.

It should also be noted that \NATCACT with DLV2 also outperforms non-hybrid query answering using DLV2 as reported in \cite{qureshi2021evaluation}, making it the fastest known method on these ontologies and queries.

\section{Discussion and Conclusion}

This work shows that the methods introduced in \cite{QureshiF23} do not only have a positive relative impact when using a hybrid reasoner, but that they can also yield the best known performance when using a suitable tool for hybrid reasoning.

It seems clear from the result that there is a benefit of keeping some portions in the ontology rather than transforming the entire ontology to facts. This is, however, contingent of the availability of a query-aware method (in this case magic sets). Among the variants, \NATCACT showed best performance, which is also the one that hybridizes most.

In the future, we plan to investigate alternative variants for producing hybrid knowledge bases and assessing their performance. Another line of future work will be to identify more hybrid reasoning systems that are query aware and benchmark these.

\bibliographystyle{eptcs}
\bibliography{reference}

\end{document}